\documentclass[11pt]{article}
\pdfoutput=1
\usepackage{graphicx}
\usepackage{comment}
\usepackage{amssymb}
\usepackage{amsmath}
\usepackage{amsthm}
\usepackage[dvipsnames]{xcolor}
\usepackage{caption, longtable, ltcaption}
\usepackage{fullpage}
\usepackage{subcaption}
\usepackage{multirow}
\usepackage{hyperref}
\usepackage{graphics}
\usepackage{tikzit}
\usepackage{tikzscale}
\usepackage[inline]{enumitem}

\tikzstyle{short vertical node}=[fill=white, draw=black, shape=rectangle, minimum height=2cm, minimum width=0.5cm]
\tikzstyle{normal vertical node}=[fill=white, draw=black, shape=rectangle, minimum height=4cm]

\tikzstyle{left dashed line for edges}=[draw=black, fill=none, dashed, <-]
\tikzstyle{right dashed line for edges}=[dashed, ->]
\tikzstyle{red edge}=[-, draw={rgb,255: red,238; green,0; blue,0}, dashed]
\tikzstyle{blue edge}=[-, draw={rgb,255: red,17; green,108; blue,255}]
\tikzstyle{directed blue edge}=[<-, draw={rgb,255: red,17; green,108; blue,255}]
\tikzstyle{directed red edge}=[<-, draw={rgb,255: red,238; green,0; blue,0}, dashed]

\usepackage{tikz}
\usetikzlibrary{positioning}
\tikzset{auto, node distance =2 cm and 3cm, on grid,
	semithick,
	dot/.style={circle,fill=black},
	state/.style ={circle, draw, black, text=black, minimum width =1cm},
	red/.style={ultra thick, color=red},
	blue/.style={ultra thick, color=blue}}

\newtheorem{theorem}{Theorem}[section]

\newtheorem{proposition}[theorem]{Proposition}

\newtheorem{lemma}[theorem]{Lemma}

\newtheorem{definition*}{Definition}[section]

\newcommand{\np}{\ensuremath{\mbox{\rm NP}}}
\newcommand{\pspace}{\ensuremath{\mbox{\rm PSPACE}}}

\mathchardef\mhyphen="2D

\newcommand{\manyone}{\ensuremath{\,\leq_{m}^{p}\,}}
\newcommand{\ncl}{\ensuremath{\mathrm{NCL}}}
\newcommand{\hanano}{\ensuremath{\mathrm{HANANO}}}
\newcommand{\andvertex}{\ensuremath{\mathrm{AND}}}
\newcommand{\orvertex}{\ensuremath{\mathrm{OR}}}

\newcommand{\bbb}{\cdots|BBB}
\newcommand{\booobb}{B\cdot\cdot|\cdot BB}
\newcommand{\obobob}{\cdot B\cdot|B\cdot B}

\newcommand{\oobbbo}{\cdot\cdot B|BB\cdot}
\newcommand{\ooo}{\cdots|BBB}
\newcommand{\boo}{B\cdot\cdot|\cdot BB}
\newcommand{\obo}{\cdot B\cdot|B\cdot B}
\newcommand{\oob}{\cdot\cdot B|BB\cdot}

\newcommand{\brr}{\cdots|BRR}
\newcommand{\booorr}{B\cdot\cdot|\cdot RR}
\newcommand{\orobor}{\cdot R\cdot|B\cdot R}
\newcommand{\oorbro}{\cdot\cdot R|BR\cdot}

\newcommand{\rbr}{\cdots|RBR}

\newcommand{\rooobr}{R\cdot\cdot|\cdot BR}
\newcommand{\oboror}{\cdot B\cdot|R\cdot R}
\newcommand{\oorrbo}{\cdot\cdot R|RB\cdot}

\newcommand{\rrb}{\cdots|RRB}
\newcommand{\rooorb}{R\cdot\cdot|\cdot RB}
\newcommand{\ororob}{\cdot R\cdot|R\cdot B}
\newcommand{\oobrro}{\cdot\cdot B|RR\cdot}
\newcommand{\screenshotwidth}{0.3\linewidth}
\newcommand{\subfigurewidth}{\linewidth}

\graphicspath{{./images/} {./images/tikz_files/} {./images/gameplay/} {./images/gadgets/}}

\newcommand{\myfbox}[2]{
	\centering
	\setlength{\fboxsep}{0pt}
	\setlength{\fboxrule}{1pt}
	\fbox{\includegraphics[width=#1\linewidth]{#2.pdf}}}

\begin{document}\sloppy
  
\title{Defying Gravity: The Complexity of the Hanano Puzzle\thanks{Work supported in part by NSF Grant CCF-2006496.}}

\author{Michael C. Chavrimootoo\\Department of Computer Science\\University of Rochester, Rochester NY 14627, USA}

\date{April 26, 2023}

\maketitle

\begin{abstract}
	Using the notion of visibility representations, our paper establishes a new property
	of instances of the Nondeterministic Constraint Logic (NCL) problem 
	(a $\pspace$-complete problem that is very 
	convenient to prove the $\pspace$-hardness of reversible games with pushing blocks).
	Direct use of this property introduces an explosion in the number of gadgets needed to show
	$\pspace$-hardness, but we show
	how to bring that number from 32 down to only three in general, and down to two in a specific case! 
	We propose it as a step towards a broader and more general framework
	for studying games with irreversible gravity, and 
	use this connection to guide an indirect polynomial-time many-one reduction from the NCL 
	problem to the Hanano Puzzle---which is $\np$-hard---to prove it is in fact 
	$\pspace$-complete.
\end{abstract}

\section{Introduction}

The application of complexity theory to the study of games has allowed us to
understand the hardness of many popular games.
Many games that are limited to a single player are 
NP-complete (with respect to many-one polynomial-time
reductions, which is what we will always refer to when using the terms 
``-hard'' and ``-complete''),
while two-player games are typically 
PSPACE-complete~\cite{dem-hea:b:games-puzzles-comp}.
However, the moment the board layout becomes dynamic or the number of moves 
becomes unbounded, the complexity of a one-player game can jump to being
PSPACE-complete~\cite{dem-hea:b:games-puzzles-comp}.
Surprisingly, the presence of irreversible gravity, which limits the number of moves possible, can yield complex games~\cite{dem-hea:b:games-puzzles-comp}.

The Hanano Puzzle is a one-player game with a dynamic board, unbounded 
moves, and gravity developed by video game creator 
Qrostar~\cite{qro:w:hanano}.
Liu and Yang recently proved that the language version of the Hanano Puzzle is 
NP-hard~\cite{liu-yan:j:hanano}. In their paper, they ask if the problem 
is NP-complete and leave the question open. We pinpoint the problem's 
complexity
by proving Hanano Puzzle's language version to be PSPACE-complete.
We do so by providing an indirect reduction from the 
 Nondeterministic Constraint Logic (NCL) problem (a known 
PSPACE-complete problem~\cite{dem-hea:b:games-puzzles-comp}).
One of the major challenges of the reduction is overcoming the effects of
gravity. We define a method that leverages graph-theoretic techniques to circumvent unwanted 
effects of gravity, thereby making reductions from NCL easier.
To our knowledge, this method (of abstracting away the ``harmful'' effects of gravity) 
is new in this area, which makes our study interesting in this sense.
We are also able to 
significantly reduce the number of gadgets that we need to build by constructing ``base 
gadgets'' from which other gadgets can be built.
This design is
entirely independent of Hanano, and so we believe it might have applications to other similar
games.

\section{Preliminaries}

A simple planar graph is one that is loop-free, has no multi-edges (i.e.,
for each pair of vertices, there is at most one edge between the vertices),
and is planar (i.e., the graph can be drawn on a piece of paper so that 
its edges only intersect at their common endpoints).

	Given a graph $G=(V, E)$, a \emph{visibility representation} $\Gamma$ for $G$ maps 
	every vertex $v \in V$ to a vertical vertex segment $\Gamma(v)$ and 
	every edge $(u, v) \in E$ to a horizontal edge segment $\Gamma(u, v)$ such that 
	each horizontal edge segment $\Gamma(u,v)$ has its respective endpoints lying on the vertical
	vertex segments $\Gamma(u)$ and $\Gamma(v)$, and no other segment intersections
	or overlaps occur~\cite{tam:b:graph-drawing}.\footnote{%
		This definition deviates slightly from the standard one in that the standard 
		definition, vertices are mapped to horizontal segments, and edges are mapped to
		vertical segments. For our purposes, both definitions are equivalent.}

\subsection{The Hanano Puzzle}\label{s:hanano-rules}

This section expands on definitions by 
Liu and Yang~\cite{liu-yan:j:hanano}.
The Hanano Puzzle comprises different levels.
A level of the game is an $n \times m$ grid (with $n, m > 0$)  that
contains only the following components:
immovable gray blocks,
movable gray blocks,
(movable) colored blocks,
colored flowers, and
empty spaces.
Each colored block/flower can be red, blue, or yellow.
Each flower is immovable and is affixed to some block.  If that block is movable, then whenever it moves, the affixed 
flower moves with the block (see Figure~\ref{f:gravity}).
Gray blocks can be of 
arbitrary shape and size, while all other components are $1 \times 1$ objects. 
In our gadgets, we try to the best of our ability to minimize the number of sides of each
movable gray block.
A block can 
slide (see Figure~\ref{f:slide})
left or right, one step at a time. For a slide to occur,
the space that the block will occupy after the slide must either be 
empty or be occupied by part of the block that is sliding.
Two adjacent blocks of width one can also be swapped in one step (see Figure~\ref{f:swap})
the positions of the two blocks can be swapped without moving any other
component of the grid.

\begin{figure}[!ht]
	\centering
\begin{subfigure}[b]{\subfigurewidth}
	\centering
	\includegraphics[width=\screenshotwidth]{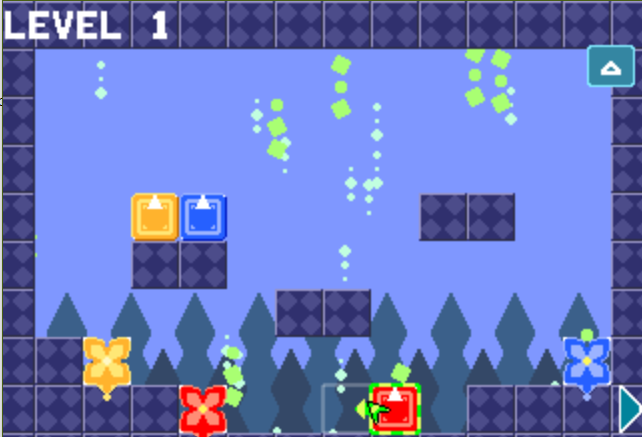}
	\includegraphics[width=\screenshotwidth]{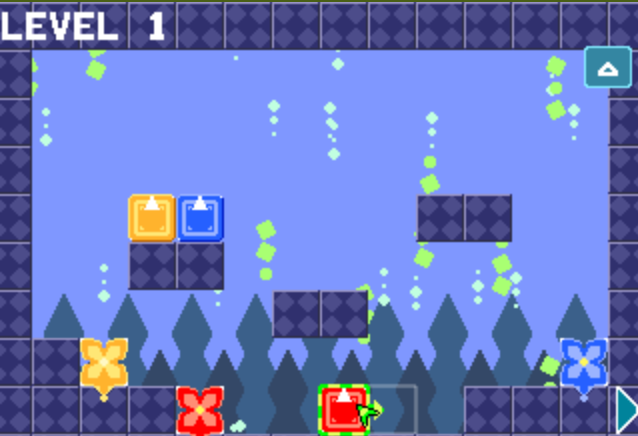}
	\caption{Example of a slide~\cite{qro:w:hanano}.}
	\label{f:slide}
\end{subfigure}
\hfill
\begin{subfigure}[b]{\subfigurewidth}
	\centering
	\includegraphics[width=\screenshotwidth]{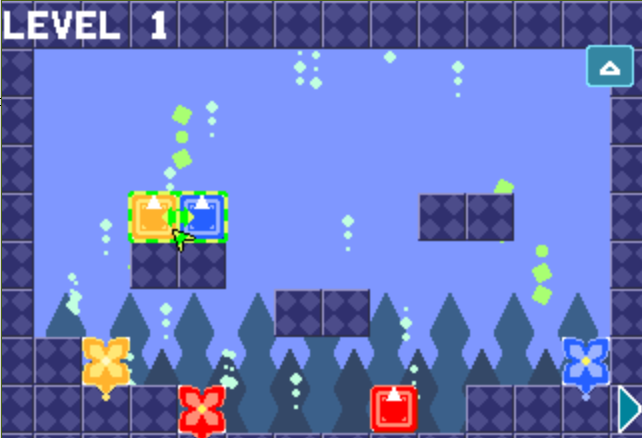}
	\includegraphics[width=\screenshotwidth]{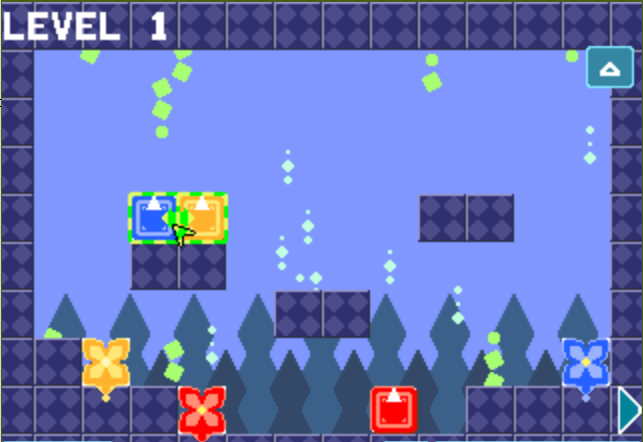}
	\caption{Example of a swap~\cite{qro:w:hanano}.}
	\label{f:swap}
\end{subfigure}
\begin{subfigure}[b]{\subfigurewidth}
	\centering
	\includegraphics[width=\screenshotwidth]{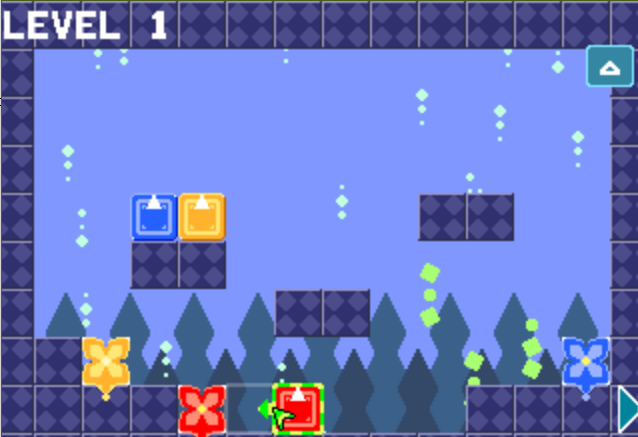}
	\includegraphics[width=\screenshotwidth]{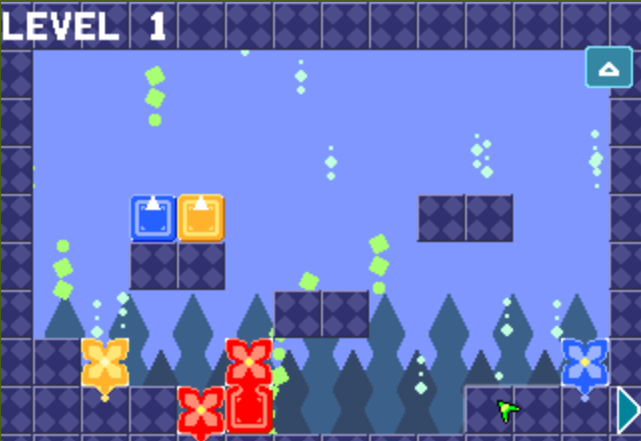}
	\caption{Example of a blooming flower~\cite{qro:w:hanano}.}
	\label{f:bloom}
\end{subfigure}
\hfill
\begin{subfigure}[b]{\subfigurewidth}
	\centering
	\includegraphics[width=\screenshotwidth]{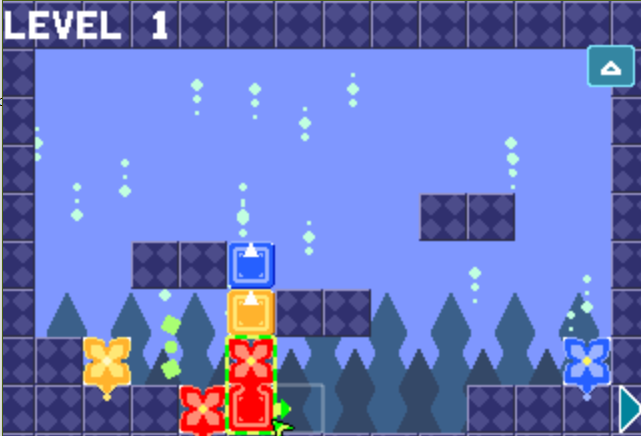}
	\includegraphics[width=\screenshotwidth]{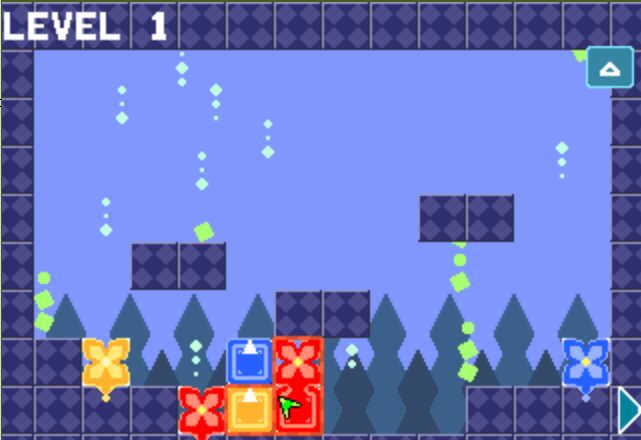}
	\caption{Example of effects of gravity~\cite{qro:w:hanano}.}
	\label{f:gravity}
\end{subfigure}
\caption{Screenshots of the Hanano Puzzle (reproduced with 
permission from Qrostar~\cite{qro:perscomm:hanano}).}
\label{f:screenshots}
\end{figure}

Figure~\ref{f:screenshots} shows screenshots of the game that show
sample game moves. Note that the 
checkered cells are what we call ``immovable gray blocks.'' Movable gray 
blocks are not depicted in these figures.
Because this is a game with gravity, after the player makes a move, every
movable block that is not directly supported will fall (see 
Figure~\ref{f:gravity}). This can be viewed as happening in a single step.
Each colored block contains an arrow, pointing either up, down, left, 
or right.
If a colored block touches (by sharing a side; touching corners have no
effect) a flower of the same color, a flower will bloom 
from the side of the colored block indicated by the arrow (see 
Figure~\ref{f:bloom}), and the new flower will stay affixed/attached to that block.
 We will sometimes say that the block has bloomed when 
this happens.
If the blooming side is in contact with a block, the 
blooming flower attempts to ``force'' its way out by pushing against the surface
in contact with the blooming side. This may result in that block in contact with the blooming
side to be shifted, or in the blooming block to be shifted.
If no shift is possible, then the flower does not bloom. 
A block can only bloom once 
and that action cannot be undone. Additionally, if the new flower is in 
contact with a different block of the same color, chain bloomings can occur 
within the same step. To solve (complete) a level, one must make every colored 
block bloom.
Formally, we determine the complexity of
$\hanano = \{H \mid H \text{ is a solvable level of the Hanano Puzzle}\}.$

\subsection{Nondeterministic Constraint Logic (NCL)} \label{s:ncl-rules}
The notions introduced in this section are from Hearn and Demaine 
\cite{dem-hea:b:games-puzzles-comp}.
An NCL graph is a directed 
graph consisting of edges of weights one or two 
(respectively called red and blue edges)
that connect vertices while subject to
the constraint that 
the sum of weights of edges into each vertex is at least two
(aka the minimum inflow requirement/constraint).
The only operation allowed on an NCL graph is flipping the direction of 
an edge such that the new graph is still an NCL graph. Given an NCL 
graph $G$ and an edge $e$ in the graph, deciding if there is a sequence of 
edge flips that eventually flip $e$ is PSPACE-complete.
It turns out that the problem remains PSPACE-complete even 
if the NCL graph is a planar AND/OR NCL graph, i.e., 
it is simple planar, each vertex is connected to exactly three edges,
and each vertex is either an AND vertex 
(i.e., one incident edge is blue and the other two are red)
or an OR vertex (i.e., all incident edges are blue).
In this paper, we will tacitly assume that our NCL graphs are planar AND/OR NCL
graphs.
For readability and accessibility purposes, in addition to being colored, the
 blue edges in this paper will be solid and the red edges will be dashed.

 Typically, to show that NCL reduces to a problem $A$, it suffices to 
construct an AND gadget and an OR gadget. However, because of the effects of 
gravity, we will see in Section~\ref{s:completeness} that the gadgets need to
``interact'' properly. 

\begin{figure}[h]
	\centering
	\begin{minipage}{0.49\linewidth}
		\ctikzfig{images/tikz_files/example_ncl}
		\caption{An NCL graph.}\label{f:ncl-example}
	\end{minipage}
	\hfill
	\begin{minipage}{0.49\linewidth}
		\ctikzfig{images/tikz_files/example_visibility_representation}
		\vspace*{-3mm}
		\caption{A visibility representation of Figure~\ref{f:ncl-example}.}\label{f:visibility-rep-example}
	\end{minipage}
\end{figure}

\section{PSPACE-Completeness} \label{s:completeness}

The upper bound below is clear and easy to see, since the configure of the game at any point in time can be stored in polynomial space.

\begin{theorem} \label{t:upper-bound}
	$\hanano \in \pspace$.
\end{theorem}

\subsection{Defying Gravity with Structure}\label{sub:structure}

Part of the difficulty in devising a correct reduction 
is the fact that in NCL, every action is fully reversible
while in HANANO, due to gravity and blooms, some moves are irreversible, so
we cannot give direct gadgets.
Liu and Yang~\cite{liu-yan:j:hanano} did not encounter this issue as 
their reduction from $\mbox{CIRCUIT-SAT}$ 
leveraged the fact that in a boolean circuit, bits only need to move in one direction once.
The technique behind their construction is very similar to that used by 
Friedman~\cite{fri:a:cubic} to show the NP-hardness of a simple game 
(Cubic) with gravity.
We thus need to be careful in our construction to make sure we do
not prematurely make irreversible moves. 
Additionally, we shall build into our
gadgets the constraints of the NCL game to help simulate NCL using
HANANO\@.

Given an NCL graph,
each node
 in the graph 
 will be simulated using gadgets. 
To help identify the colored blocks and flowers of the gadgets in proofs, 
we label those items with special text.
E.g., the label ``B2'', indicates the second blue block in the gadget,
whereas the label ``BF1'' indicates the first blue flower.
Blue blocks will represent both red and blue edges. 
It's important to note that \emph{all our blocks bloom upwards}.
Next to each flower will be a boldfaced white line to indicate where the flower is attached.
Additionally, our gadgets will contain some grid lines to help the reader better gauge the distances.
The presence of a block in a gadget will indicate that the edge
represented by the block is directed into the node represented by the gadget.
Blocks will only be allowed to move between gadgets by following the 
constraints imposed by NCL\@. This means that those constraints must be 
encoded within the gadgets, using the rules of the Hanano Puzzle.
This is the first challenge.
The second challenge is to overcome the nonreversibility induced by gravity.
To help ensure that most block moves are 
reversible, 
the effects of gravity must be circumvented. Luckily, this is 
possible due to the planarity of NCL graphs. We shall start by addressing 
the second challenge. The first challenge will be resolved by designing the 
gadgets.
One such way is by having blooms ``force'' a certain setting of the 
game 
when the inflow constraints are violated.

\begin{theorem}[\cite{tam-tol:j:visibility-representations,tam:b:graph-drawing}]
A graph admits a visibility representation if and only if it is planar. Furthermore,
a visibility representation for a planar graph can be constructed in linear time.\footnote{
	In an earlier version of this paper, we independently proved a weaker version of this theorem. 
	Our result established the ``if'' direction (the ``only if'' direction is trivial), 
	but our polynomial-time algorithm did not run in linear time.}
\end{theorem}

Since our NCL graphs are planar, we can compute the 
visibility representation of an NCL graph in linear time. This has the advantage that if we're 
trying to reduce to a game of sliding blocks, and we wish to represent the direction of the
edges by using blocks that are sliding from one gadget (where each vertex is represented using a 
gadget)
to another, then by having all the edges be
horizontal, we remove the danger of having gravity make an ``edge flip'' irreversible. 
Implicit in this, is that our gadgets will need to be ``size-adaptable,'' i.e., as we will see
based on the length of the segment to which a vertex is mapped, a gadget's height may need to 
change. Our constructions will have the property that they can be internally padded so as to make
gadgets artificially long without affecting their correctness.
Thus those
``edge flips'' in the game of interest (here, the Hanano Puzzle) are fully reversible, to the extent that
is required to be compatible with NCL's reversibility. 

\subsection{Gadgets and Schemas}

Each vertex in the NCL graph will be represented using a gadget and
those gadgets will be connected using tunnels that will represent the edges.
For each of these tunnels, there will be a blue block and that block will be placed in the gadget
representing the vertex to which the edge is incident. Thus, flipping edges
will be represented by moving blocks from one gadget to another. For gadgets
to interact properly, we must ensure that a block can only travel through its 
designated tunnel and that the minimum inflow requirement is always met,
i.e., for each gadget, there is always either one blue block representing a blue edge or two blue blocks each representing a red edge
in the gadget, at any point in time.
Since each vertex in the NCL graph is connected to exactly three edges, 
each gadget will have three entry points that can each lie either on the 
left or the right of the gadget. 
Consider the following notation to represent a gadget: 
$x_1x_2x_3|y_1y_2y_3$, where for each $i \in \{1, 2, 3\}$, $\{x_i, y_i\} \in 
\{\{R, \cdot\}, \{B, \cdot\}\}$, and the list $[x_1, x_2, x_3, y_1, y_2, y_3]$ contains either
exactly
 three $B$s, or exactly one $B$ and two $R$s. (The last condition simply captures the idea that each vertex is either an AND vertex or an OR vertex.)
For example, $\rooobr$ means that the top 
entry point on the left side of the gadget is for a red edge's blue block, 
the middle on the right side is for a blue edge's blue block, and the bottom on the right
side is for a red edge's blue block. 
The remaining entry points are considered blocked off, i.e., not entry 
points. 
It's easy to see that due to this structure enforced, the number of gadgets 
needed to show the reduction goes up from two to~\emph{32} (8 OR gadgets~+ 24 
AND gadgets).\footnote{%
	One could try to argue that given one OR gadget and one AND gadget, it suffices to ``just place the tunnels on the correct side'' to obtain the remaining~30 gadgets, but that approach does not take into account the structure of the gadget. This suggested approach certainly works for our OR gadgets (see Figure~\ref{f:or-gadget}), but not for our AND gadgets (see Figure~\ref{f:and-gadget}). However, we do show how to derive all the OR/AND gadgets in a way that is independent of the structure of the underlying gadgets.%
}
However, we will show how, by giving only three gadgets, we can derive
all the remaining gadgets (i.e., show their existence). The approach will in fact not rely
on properties of the Hanano Puzzle, and thus will be reusable for other purposes.
Moreover, if one restricts their attention to HANANO, then we will argue that one of the gadgets need not be constructed.
Let us first look at how to construct the OR gadgets. 

\newcommand{\biggadgetwidth}{0.82}

\begin{figure}[h!]
	\centering
	\begin{subfigure}[b]{0.49\linewidth}
		\myfbox{\biggadgetwidth}{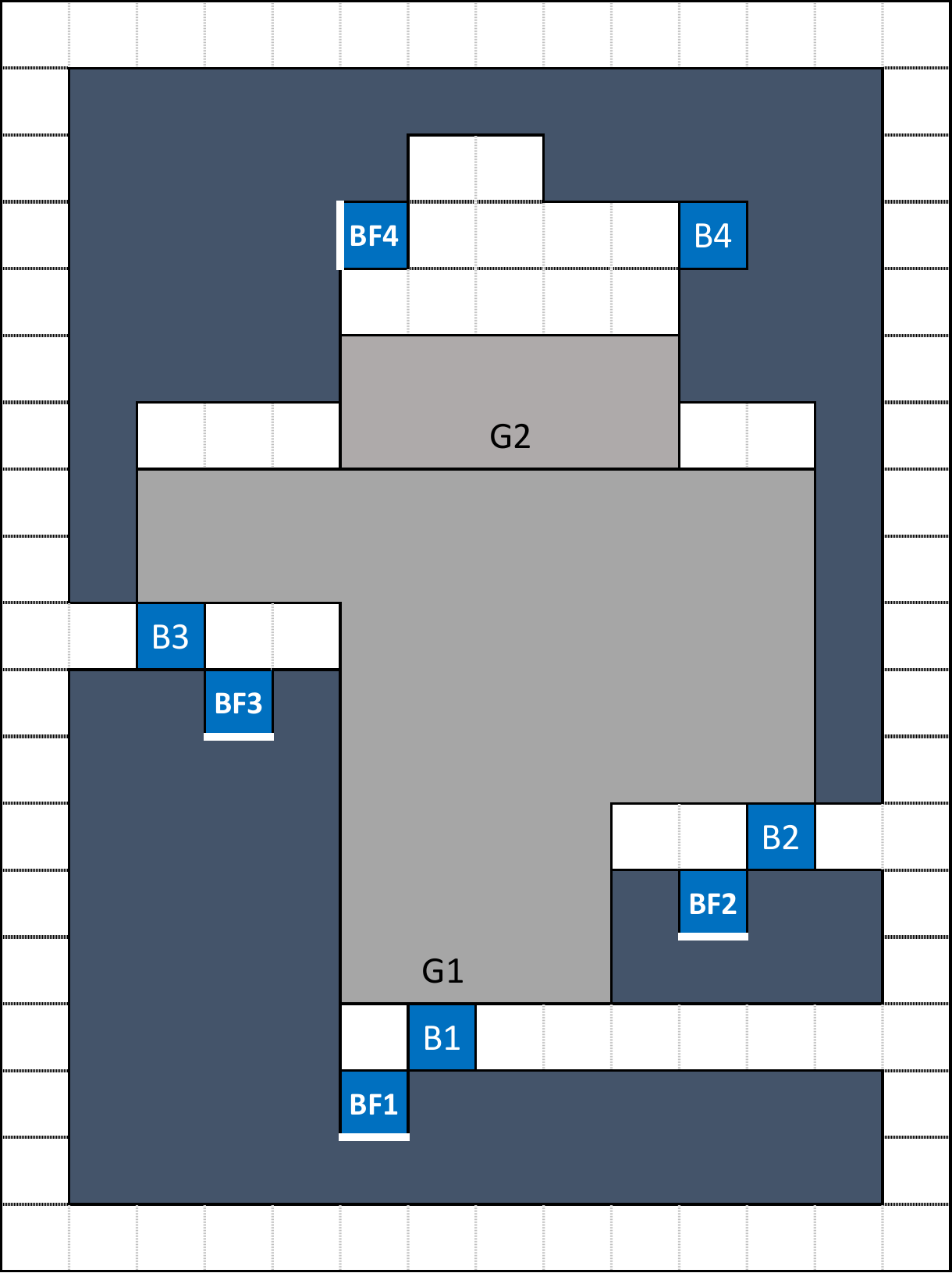}
		\caption{$\booobb$ gadget.}\label{f:or-gadget}
	\end{subfigure}
	\hfill
	\begin{subfigure}[b]{0.49\linewidth}
		\myfbox{\biggadgetwidth}{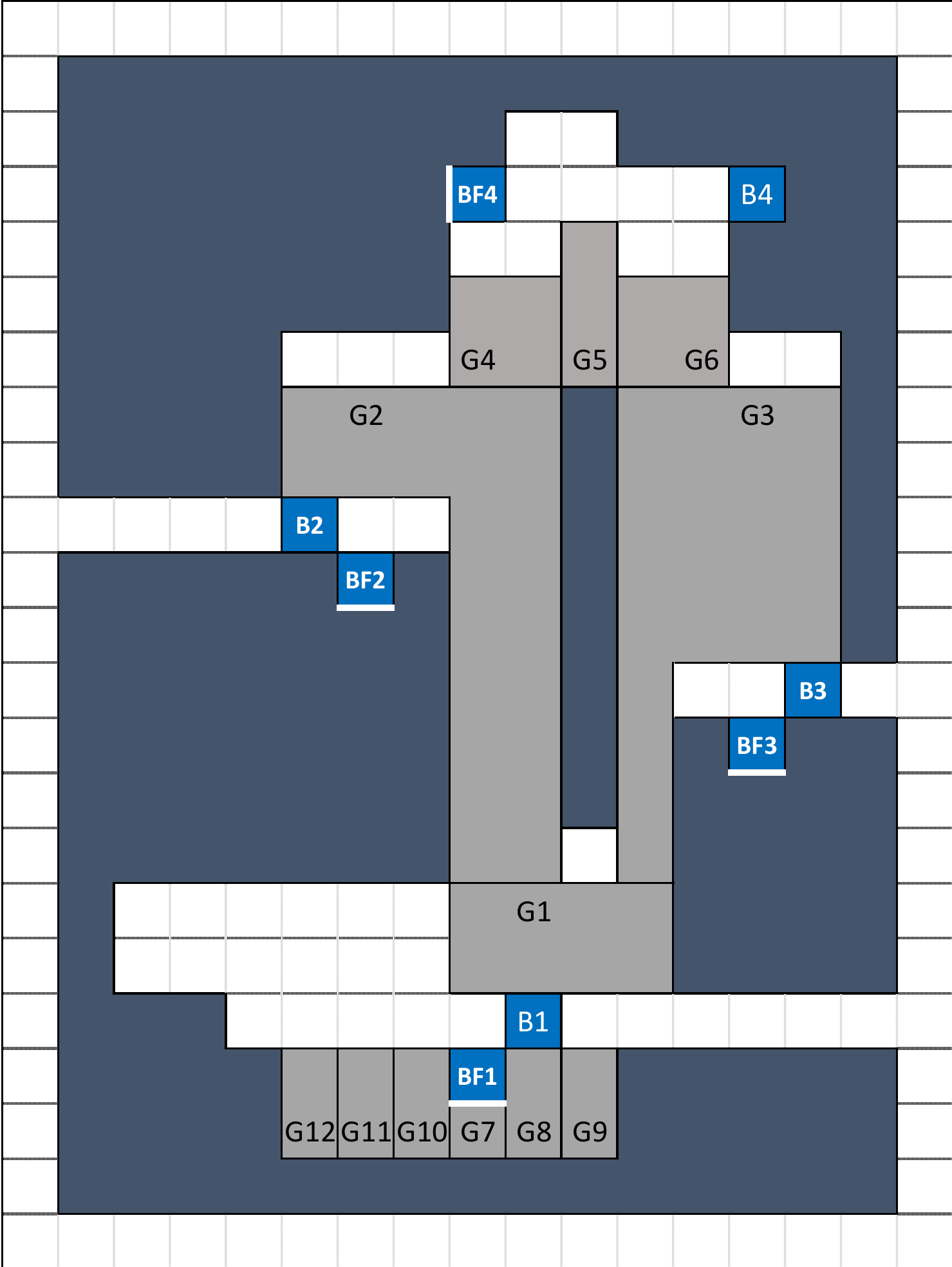}
		\caption{$\rooorb$ gadget.}\label{f:and-gadget}
	\end{subfigure}
	\hfill
	\begin{subfigure}[b]{0.49\linewidth}
		\myfbox{\biggadgetwidth}{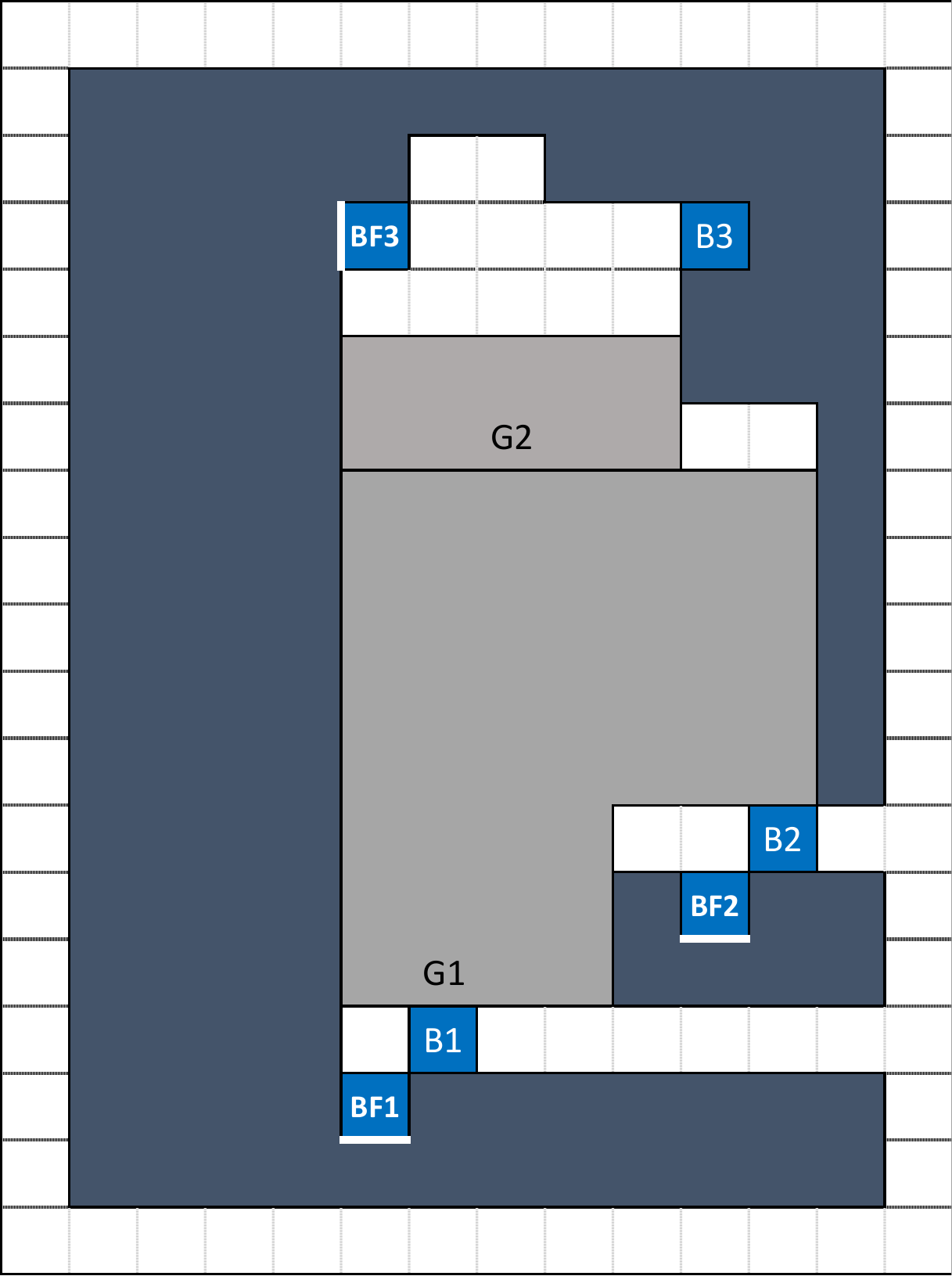}
		\caption{Red bend gadget.}\label{f:bend}
	\end{subfigure}
	\hfill
	\caption{Our three gadgets: an OR gadget ($\booobb$), an AND gadget ($\rooorb$), and a red bend gadget.}\label{f:all-gadgets}
\end{figure}

\begin{lemma}\label{l:or-correctness}
	The gadget in Figure~\ref{f:or-gadget} satisfies the same constraints as an $\ncl$ $\orvertex$ vertex.
\end{lemma}
\begin{proof}
	First notice that each movable gray block has very limited movement.
	G2 can only move up by one ``unit,'' and G1 can either move up or move down
	by one unit. Thus for any blue block B$x$, the only flower that it can reach in
	that gadget is BF$x$.
	Now, the only way for B4 to bloom is if B4 is in contact with BF4, and it must be
	on BF4's right side (the only other exposed side of BF4 is the bottom side, but if B4 is 
	directly under BF4 it will not have enough room to bloom). 
	Thus B4 can bloom iff G2 moves up by one unit. This can happen exactly if G1
	moves up by one unit, which can happen iff one of B1, B2, or B3 blooms. 
	Finally, notice that if B1, B2, and B3 all leave the gadget, then G1 and G2 both drop
	by one unit with no possibility of returning to their original configuration, thus
	making it impossible to bloom B4.
	
	We conclude by noting that we could have merged G1 and G2 into a single block, 
	but opted not to as we sought to minimize the number of sides on each movable gray block.
\end{proof}

We will show how to construct certain gadgets from other gadgets, by essentially chaining certain gadgets together.
For our convenience we define a ``constrained blue edge terminator'' gadget, that will allow us to force an edge from pointing out of a gadget, without connecting the edge to other nodes. This allows us to simplify the design of our gadgets.
We state the following proposition in a general form, i.e., its proof will not depend on the
Hanano Puzzle's properties.

\begin{proposition}\label{p:constrained_blue}
	The constrained blue edge terminator gadget can be constructed using any 
	gadget that satisfies the same constraints as an $\ncl$ $\orvertex$ vertex.
\end{proposition}
\begin{proof}
	We want the constrained blue edge terminator to be a gadget that, when attached to a tunnel that represents a blue edge, will force the block that represents the edge's orientation to be inside itself (i.e., inside the constrained blue edge terminator) so as to not violate the inflow constraints.

	Fix a gadget that satisfies the same constraints as an NCL OR vertex. There must exist a configuration of the gadget that corresponds to the NCL OR 
	with one (blue) edge pointing into the vertex and the remaining edges (i.e., two blue edges) 
	pointing out of the vertex. Now, block off the tunnels that correspond to the two
	edges pointing of the vertex. This configuration of the gadget 
	correctly constraints the represented blue edge's orientation.
\end{proof}

A natural question to ask is whether $\booobb$ is special, or whether this result can be achieved
using any of the other OR gadgets, and we answer in the positive that indeed, any 
of the eight OR gadgets suffices.
We first note that it suffices to consider the gadgets for $\boo$, $\obo$, $\oob$, and $\ooo$,
as the remaining ones can be obtained via vertical symmetry.
In an abuse of notation, we will sometimes use the shorthand for a gadget to refer
to the gadget that is obtained from it via vertical symmetry in our schemas.
Edges that are connected to the constrained blue edge terminator have a direction assigned and
point to a $\oslash$ to indicate the termination. The remaining edges have no direction, indicating
that they can be assigned in any way that satisfies the minimum inflow constraints. 

\begin{lemma}\label{l:or-equivalence}
	For each gadget $G$ in the following list, the remaining gadgets in that same list can be constructed
	from $G$: 
	\begin{enumerate*}[itemjoin={{, }}, itemjoin*={{, and }}, after={{.}}]
		\item\label{or:one} $\booobb$
		\item\label{or:two} $\obobob$
		\item\label{or:three} $\oobbbo$
		\item\label{or:four} $\bbb$
	\end{enumerate*}
\end{lemma}
\begin{proof}%
	We prove this lemma by showing how to construct~\ref{or:two} from~\ref{or:one}, how to
	 construct~\ref{or:three} from~\ref{or:two},
	how to construct~\ref{or:four} from~\ref{or:three}, and finally how to 
	construct~\ref{or:one} from~\ref{or:four}.
	In each case, we will have a gadget that satisfies the same constraints as an NCL OR vertex, so we
	tacitly appeal to Proposition~\ref{p:constrained_blue} to, for ``free,'' have a constrained blue
	edge terminator. We construct in Figure~\ref{f:or-schemas} schematic diagrams to aid in our proof.

	From~\ref{or:one} to~\ref{or:two}. Figure~\ref{f:schema-obo} depicts the construction.
	The edges of the $\obo$ gadget are~2, 4, and~5.
	If those edges all
	point out (i.e.,~2 points left and the other two point right), then the minimum inflow constraint
	is violated as~3 cannot be flipped and~1 can only point into one the two gadgets. Thus one 
	of~2, 4, or~5 must always be pointed inwards, and this gadgets satisfies the same constraints and as NCL OR vertex.

	The rest of the proof is analogous to the first case, and we provide it for completeness.

	From~\ref{or:two} to~\ref{or:three}. Figure~\ref{f:schema-oob} depicts the construction.
	The edges of the $\obo$ gadget are~2, 4, and~5.
	If those edges all
	point out (i.e.,~2 points left and the other two point right), then the minimum inflow constraint
	is violated as~3 cannot be flipped and~1 can only point into one the two gadgets. Thus one 
	of~2, 4, or~5 must always be pointed inwards, and this gadgets satisfies the same constraints and as NCL OR vertex.

	From~\ref{or:three} to~\ref{or:four}. Figure~\ref{f:schema-ooo} depicts the construction.
	The edges of the $\obo$ gadget are~2, 4, and~5.
	If those edges all
	point out (i.e., they all point right), then the minimum inflow constraint
	is violated as~3 cannot be flipped and~1 can only point into one the two gadgets. Thus one 
	of~2, 4, or~5 must always be pointed inwards, and this gadgets satisfies the same constraints and as NCL OR vertex.

	From~\ref{or:four} to~\ref{or:one}. Figure~\ref{f:schema-boo} depicts the construction.
	The edges of the $\obo$ gadget are~2, 4, and~5.
	If those edges all
	point out (i.e.,~2 points left and the other two point right), then the minimum inflow constraint
	is violated as~3 cannot be flipped and~1 can only point into one the two gadgets. Thus one 
	of~2, 4, or~5 must always be pointed inwards, and this gadgets satisfies the same constraints and as NCL OR vertex.
\end{proof}

\begin{figure}[!ht]
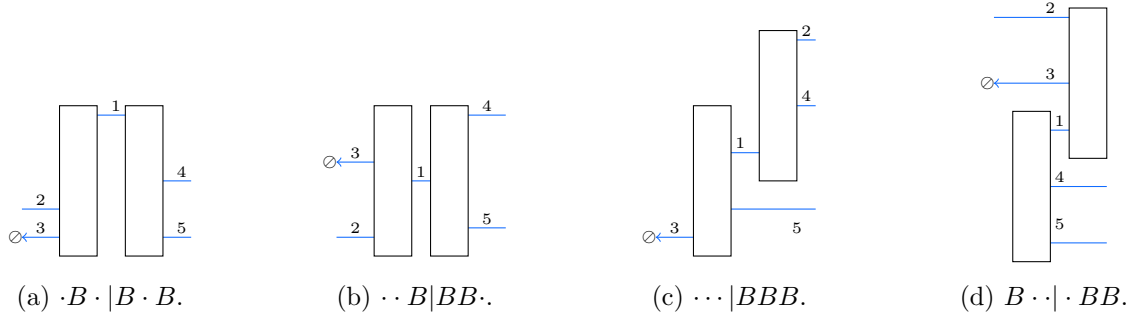

	\centering
	\begin{subfigure}[b]{0.24\linewidth}
		\ctikzfig{images/tikz_files/or_schemas/obo}
		\vspace*{-25mm}
		\caption{$\obo$.}\label{f:schema-obo}%
	\end{subfigure}
	\hfill
	\begin{subfigure}[b]{0.24\linewidth}
		\ctikzfig{images/tikz_files/or_schemas/oob}
		\vspace*{-25mm}
		\caption{$\oob$.}\label{f:schema-oob}%
	\end{subfigure}
	\hfill
	\begin{subfigure}[b]{0.24\linewidth}
		\ctikzfig{images/tikz_files/or_schemas/ooo}
		\vspace*{-15mm}
		\caption{$\ooo$.}\label{f:schema-ooo}%
	\end{subfigure}
	\hfill
	\begin{subfigure}[b]{0.24\linewidth}
		\ctikzfig{images/tikz_files/or_schemas/boo}
		\vspace*{-12mm}
		\caption{$\boo$.}\label{f:schema-boo}%
	\end{subfigure}
	\caption{Schemas showing the ``equivalence'' of the OR gadgets.}\label{f:or-schemas}
\end{figure}

The AND gadgets are trickier, but we can construct all of them using two gadgets: 
a gadget that satisfies the same constraints as an NCL AND vertex, and a ``red bend'' gadget.
The red bend gadget will allow us to, in some sense, reorient the direction of a tunnel that corresponds to a red edge (hence the ``red'' in the name even though all our blocks are blue; in the context of HANANO, the red bend gadget also works on ``blue edges'' but that does not necessarily hold in general).
Lemma~\ref{l:and-equivalence} gives an analogous result to Lemma~\ref{l:or-equivalence}. 

\begin{lemma}\label{l:and-correctness}
The gadget in Figure~\ref{f:and-gadget} satisfies the same constraints as an $\ncl$ $\andvertex$ vertex, with B1 representing the blue edge, and with B3 and B4 representing the red edges.
\end{lemma}
\begin{proof}
	Let us first describe the gadget before arguing its correctness.

	G4 and G6 can only move up by one ``unit,'' and 
	G2 and G3 can each either move up or move down
	by one unit. Thus each colored block, can only reach one flower.
	If both B2 and B3 exit the gadget, then B1 must remain to support G1 (which in 
	turn supports G2 and G3), as otherwise, G2 and G3 will drop by one unit 
	and G4 and G6 will never be able to move up.
	Similarly, if B1 is to exit, all the gray blocks must remain supported. 
	This is only possible if G1 is stowed to the left and B2 and B3
	remain in the gadget to, respectively, support G3 and G2. 
	Additionally, the area underneath B1 is made up of multiple movable gray blocks for a 
	simple reason. B1 must be able to move horizontally without blooming (either to exit
	the gadget or two carry G1 to ``stow'' it). By this setting, we can move
	the location of BF1 (by swapping G7 with an adjacent movable block of width one; BF1 is affixed to G7 and the two will move as if they were one $2 \times 1$ block) to make sure it is always on the left of B1.

	Now, the only way for B4 to bloom is if B4 is in contact with BF4, and it must be
	on BF4's right side (the only other exposed side of BF4 is the bottom side, but if B4 is 
	directly under BF4 it will not have enough room to bloom). 
	Thus B4 can bloom if and only if G4 and G6 move up by one unit. This can happen exactly if both
	G2 and G3 move up by one unit. There are two ways this can happen:
	Either both B2 and B3 bloom, or B1 blooms (thus pushing G1 up by one unit).
	Thus B4 blooms if and only if either B2 and B3 bloom in the gadget, or B1 blooms in the gadgets.
\end{proof}

We now define an important property of the red bend gadget that is used in the proof of Lemma~\ref{l:and-equivalence}.	Intuitively, the results means that the 
inflow constraint on red bend gadgets is one (or two under the ``blue bend'' interpretation).

\begin{proposition}\label{p:bend}
	The red bend gadget in Figure~\ref{f:bend} is solvable if and only if either B1 or B2 blooms while supporting G1.
\end{proposition}
\begin{proof}
	The design of the gadget is simply a restricted/modified version of that in Figure~\ref{f:or-gadget},
	so we omit the description of the gadget. 

	$\implies$: Suppose the gadget is solvable. Then B3 must come in contact with BF3. This is only possible if both G1 and G2 move up by one unit. For this to happen, either B1 or B2 must bloom while supporting G1.

	$\impliedby$: Suppose that either B1 or B2 blooms while supporting G1. In both cases, G1 moves up by one
	unit, thus also pushing G2 up by one unit, hence allowing B3 to come in contact with BF3 to bloom.
\end{proof}

Earlier, we mentioned that when our attention is focus on HANANO, we only need two gadgets. This is indeed possible because all the blocks that we use are of the same color, and so we can define the red bend gadget to be a restricted version of the OR gadget. Consider the gadget in Figure~\ref{f:or-gadget}. If we place a constrained blue edge terminator at the tunnel for B3, then the resulting gadget is essentially the red bend gadget.

\begin{lemma}\label{l:and-equivalence}
	For each gadget $G$ in the following list, the remaining gadgets in that same list can be constructed
	from $G$, the red bend gadget, and any $\orvertex$ gadget:
	\begin{enumerate*}[itemjoin={{, }}, itemjoin*={{, and }}, after={{.}}]
		\item\label{and:one} $\rooorb$
		\item\label{and:two} $\ororob$
		\item\label{and:three} $\rrb$
		\item\label{and:four} $\oobrro$
		\item\label{and:five} $\brr$
		\item\label{and:six} $\oorbro$
		\item\label{and:seven} $\booorr$
		\item\label{and:eight} $\orobor$
		\item\label{and:nine} $\rbr$
		\item\label{and:ten} $\rooobr$
		\item\label{and:eleven} $\oboror$
		\item\label{and:twelve} $\oorrbo$
	\end{enumerate*}
\end{lemma}
\begin{proof}%
	We represent our red bend gadgets using a vertex with exactly two
	red edges on the same side of the gadget. The structure of this proof resembles that of Lemma~\ref{l:or-equivalence}.

	From~\ref{and:one} to~\ref{and:two}. Figure~\ref{f:schema-down-oro}
	depicts the construction. If edge~5 points right, then edge~1 points 
	right and edge~4 points left. Thus edge two must point left, leaving edge~3
	to point right. If edge~5 points left, then edge~4 is free to point in either
	direction. In that case, we can fix edge~1 to point left and edge~2 to point
	right, thus leaving edge~3 to point in any direction.
	Therefore, this gadget satisfies the same constraints as an NCL AND vertex.

	In the rest of this proof, whenever we use a number $x$, we implicitly mean
	``edge~$x$.''

	From~\ref{and:two} to~\ref{and:three}. Figure~\ref{f:schema-down-ooo} depicts
	the construction. If~4 points right, then~2 must point right and~3 must point left.
	Thus~1 must point left.
	If~4 points left, we can fix~2 to point left, leaving~1 and~3 free to 
	point in either direction. Therefore, this gadget satisfies the same constraints as an NCL AND vertex.

	From~\ref{and:three} to~\ref{and:four}. Figure~\ref{f:schema-down-oob} depicts
	the construction. 
	If~4 points left, then~3 must point right, forcing both~1 and~2 to point left.
	If~4 points right, we can fix~3 to point left, thus leaving~1 and~2 free to
	point in either direction.
	Therefore, this gadget satisfies the same constraints as an NCL AND vertex.

	From~\ref{and:four} to~\ref{and:five}. Figure~\ref{f:schema-top-ooo} depicts
	the construction. 
	If~4 points right, then~3 must point left, forcing both~1 and~2 to point left.
	If~4 points left, we can fix~3 to point right, thus leaving~1 and~2 free to
	point in either direction.
	Therefore, this gadget satisfies the same constraints as an NCL AND vertex.

	From~\ref{and:five} to~\ref{and:six}. Figure~\ref{f:schema-top-oor} depicts
	the construction. 
	If~1 points right, then~2 and~3 must point left, and so~4 must point left.
	If~1 points left, then we can fix~3 to point right, leaving~2 and~4 free
	to point in either direction.
	Therefore, this gadget satisfies the same constraints as an NCL AND vertex.

	From~\ref{and:six} to~\ref{and:seven}. Figure~\ref{f:schema-top-boo} depicts
	the construction. 
	If~1 points right, then~2 must point left and~3 must point right, 
	and so~4 must point right.
	If~1 points left, then we can fix~2 to point right, leaving~3 and~4 free
	to point in either direction.
	Therefore, this gadget satisfies the same constraints as an NCL AND vertex.

	From~\ref{and:seven} to~\ref{and:eight}. Figure~\ref{f:schema-top-oro} depicts
	the construction. 
	If~1 points right, then~2 and~3 must point right too, and so~4 must point left.
	If~1 points left, then we can fix~3 to point left, 	
	leaving~2 and~4 free
	to point in either direction.
	Therefore, this gadget satisfies the same constraints as an NCL AND vertex.

	From~\ref{and:eight} to~\ref{and:nine}. Figure~\ref{f:schema-mid-ooo} depicts
	the construction. 
	If~1 points right, then~2 must point right and~3 must point left, 
	and so~4 must point left.
	If~1 points left, then we can fix~2 to point left, 	
	leaving~3 and~4 free
	to point in either direction.
	Therefore, this gadget satisfies the same constraints as an NCL AND vertex.

	From~\ref{and:nine} to~\ref{and:ten}. Figure~\ref{f:schema-mid-roo} depicts
	the construction. 
	If~1 points right, then~2 and~3 must point left, 
	and so~4 must point right.
	If~1 points left, then we can fix~3 to point right, 	
	leaving~2 and~4 free
	to point in either direction.
	Therefore, this gadget satisfies the same constraints as an NCL AND vertex.

	From~\ref{and:ten} to~\ref{and:eleven}. Figure~\ref{f:schema-mid-obo} depicts
	the construction. 
	If~1 points right, then~2 must point left and~3 must point right, 
	and so~4 must point right.
	If~1 points left, then we can fix~2 to point right, 	
	leaving~3 and~4 free
	to point in either direction.
	Therefore, this gadget satisfies the same constraints as an NCL AND vertex.

	From~\ref{and:eleven} to~\ref{and:twelve}. Figure~\ref{f:schema-mid-oor} depicts
	the construction. 
	If~1 points right, then~2 and~3 must point right, 
	and so~4 must point left.
	If~1 points left, then we can fix~3 to point left, 	
	leaving~2 and~4 free
	to point in either direction.
	Therefore, this gadget satisfies the same constraints as an NCL AND vertex.

	From~\ref{and:twelve} to~\ref{and:one}. Figure~\ref{f:schema-down-roo} depicts
	the construction. 
	If~1 points right, then~2 must point right and~3 must point left, 
	and so~4 must point left, and~5 must point right.
	If~1 points left, then we can fix~2 to point left and fix~4 to point right, 	
	leaving~3 and~5 free
	to point in either direction.
	Therefore, this gadget satisfies the same constraints as an NCL AND vertex.
\end{proof}

\newcommand{\andschemaswidth}{0.3\linewidth}

\begin{figure}[h!]
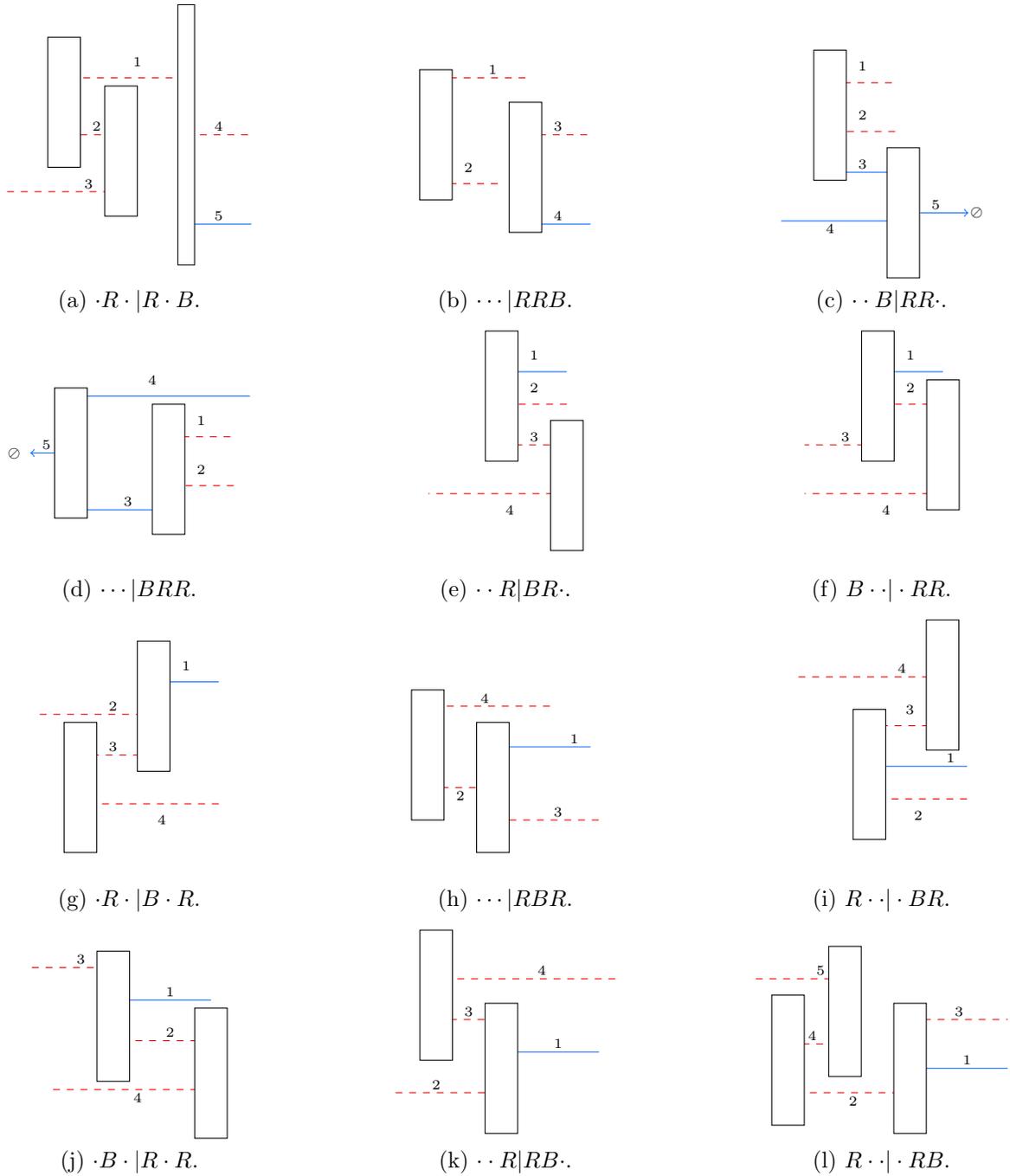

	\centering
	\begin{subfigure}[b]{\andschemaswidth}
		\ctikzfig{images/tikz_files/and_schemas/down-oro}
		\vspace*{-15mm}
		\caption{$\ororob$.}\label{f:schema-down-oro}%
	\end{subfigure}
	\hfill
	\begin{subfigure}[b]{\andschemaswidth}
		\ctikzfig{images/tikz_files/and_schemas/down-ooo}
		\vspace*{-15mm}
		\caption{$\rrb$.}\label{f:schema-down-ooo}%
	\end{subfigure}
	\hfill
	\begin{subfigure}[b]{\andschemaswidth}
		\ctikzfig{images/tikz_files/and_schemas/down-oob}
		\vspace*{-7mm}
		\caption{$\oobrro$.}\label{f:schema-down-oob}%
	\end{subfigure}
	\begin{subfigure}[b]{\andschemaswidth}
		\ctikzfig{images/tikz_files/and_schemas/top-ooo}
		\vspace*{-17mm}
		\caption{$\brr$.}\label{f:schema-top-ooo}%
	\end{subfigure}
	\hfill
	\begin{subfigure}[b]{\andschemaswidth}
		\vspace{2mm}
		\ctikzfig{images/tikz_files/and_schemas/top-oor}
		\vspace*{-7mm}
		\caption{$\oorbro$.}\label{f:schema-top-oor}%
	\end{subfigure}
	\hfill
	\begin{subfigure}[b]{\andschemaswidth}
		\ctikzfig{images/tikz_files/and_schemas/top-boo}
		\vspace*{-7mm}
		\caption{$\booorr$.}\label{f:schema-top-boo}%
	\end{subfigure}
	\begin{subfigure}[b]{\andschemaswidth}
		\ctikzfig{images/tikz_files/and_schemas/top-oro}
		\vspace*{-7mm}
		\caption{$\orobor$.}\label{f:schema-top-oro}%
	\end{subfigure}
	\hfill
	\begin{subfigure}[b]{\andschemaswidth}
		\ctikzfig{images/tikz_files/and_schemas/mid-ooo}
		\vspace*{-17mm}
		\caption{$\rbr$.}\label{f:schema-mid-ooo}%
	\end{subfigure}
	\hfill
	\begin{subfigure}[b]{\andschemaswidth}
		\vspace*{2mm}
		\ctikzfig{images/tikz_files/and_schemas/mid-roo}
		\vspace*{-20mm}
		\caption{$\rooobr$.}\label{f:schema-mid-roo}%
	\end{subfigure}
	\hfill
	\begin{subfigure}[b]{\andschemaswidth}
		\ctikzfig{images/tikz_files/and_schemas/mid-obo}
		\vspace*{-17mm}
		\caption{$\oboror$.}\label{f:schema-mid-obo}%
	\end{subfigure}
	\hfill
	\begin{subfigure}[b]{\andschemaswidth}
		\vspace*{2mm}
		\ctikzfig{images/tikz_files/and_schemas/mid-oor}
		\vspace*{-25mm}
		\caption{$\oorrbo$.}\label{f:schema-mid-oor}%
	\end{subfigure}
	\hfill
	\begin{subfigure}[b]{\andschemaswidth}
		\ctikzfig{images/tikz_files/and_schemas/down-roo}
		\vspace*{-25mm}
		\caption{$\rooorb$.}\label{f:schema-down-roo}%
	\end{subfigure}
	\caption{Schemas showing the ``equivalence'' of the AND gadgets.}\label{f:and-schemas}
\end{figure}

\subsection{Main Result}

\begin{theorem}\label{t:pspace-complete}
	$\hanano$ is $\pspace$-complete even if
	(1)~all flowers and colored blocks have the same color, and (2)~colored blocks can only bloom
	upwards.
\end{theorem}
\begin{proof}
	Since $\hanano \in\pspace$, it suffices to show that $\ncl\manyone\hanano$.
	Consider the clearly-polynomial-time-computable function that we describe in the next paragraph. 
	We assume without loss of generality that the input
	is an NCL graph and a valid target edge, as we can easily detect in polynomial time if is not and map to a fixed
	element that is not in HANANO\@.
	
	Construct in polynomial time a visibility representation for the input NCL graph
	and construct a game grid based on the visibility representation, replacing each vertex of 
	the graph by a suitable gadget, and replacing
	edges with the appropriate tunnels.
	The game grid will be 
	polynomially larger than the visibility representation since the gadgets have 
	constant-bounded size. 
	We must ensure that the game is only solvable when the target edge 
	$e = (u, v)$ is flipped. 
	If the flower than blooms $b$ is attached to an immovable gray block, replace
	that flower with an immovable gray block.
	Otherwise, the flower that blooms $b$ must be attached to the top of a $1\times1$ 
	movable gray block. Replace that gray block with a $2\times 1$ movable gray block.
	There is now no flower in the gadget for $v$ that can bloom $b$, so to bloom, $b$ must move the gadget for $u$.
	If the game is solvable, then $b$ must bloom, and so there will exist a sequence 
	of block movements corresponding to edge flips, so the edge $e$ can be flipped in $G$. 
	If there is a sequence of edge flips that eventually flips edge $e$ in $G$, there is sequence of block movements that respect the inflow constraints and eventually see $b$ move from the gadget representing $v$ to the gadget representing $u$. Thus the colored block $b$ (and all the other ones in the game) can bloom and the game is solvable. 
	Finally note that all the colored blocks in our gadgets have their arrows
	pointing up, and that we only use blue blocks/flowers.
\end{proof}

\section{Related Work}\label{s:related-work}
The literature on the complexity of games is 
rich and covers a variety 
of games.
For general earlier results, we refer readers to Appendix~A of Hearn and 
Demaine's book~\cite{dem-hea:b:games-puzzles-comp}, which contains an 
extensive survey of games whose complexity were known at their time of 
writing. 

The introduction of NCL in an early paper by Hearn and Demaine~\cite{dem-hea:j:ncl-pspace}
helped simplify the process of showing that many games with sliding blocks 
are PSPACE-complete by limiting the number of gadgets to simulate to two.
The work on motion planning through doors~\cite{ani-bos-dem-dio-hen-lyn:c:door-pspace-hard}
provides a framework to show the PSPACE-hardness of certain problems by 
simulating \emph{one} gadget. 
However, that paper's contribution does not solve the major problem faced by 
classifying the Hanano Puzzle: circumventing certain effects of gravity.
There are games with gravity that were studied prior to the 
introduction of NCL\@.
For example, Friedman~\cite{fri:a:cubic} proved Cubic to be 
NP-hard using a similar construction to that of Liu and Yang~\cite{liu-yan:j:hanano}.\footnote{
	The paper actually claims to show the NP-completeness of Cubic. 
	However, there is no apparent proof (or even claim) in the text of a 
	matching upper	bound.}
Clickomania is another game with gravity that was studied before the 
introduction of NCL\@. This game is a one-player game with
a bounded number of moves and it is in fact NP-complete~\cite{bie-dem-dem-fle-jac-mun:b:clickomania}.
Solving a level of Super Mario Brothers (SMB), which is another game with gravity, has also been proven to be 
PSPACE-complete~\cite{dem-vig-wil:c:mario-bros}. 
However, the framework used in that proof does not
rely on NCL, since SMB is not a game that involves pulling blocks.
Another famous game with gravity is Tetris. While the ``offline'' version
is NP-complete~\cite{bre-dem-hoh-hoo-kos-lib:j:tetris-hard}, 
in the general case, it is NP-hard~\cite{asi-cou-dem-dem-hes-lyn-sin:j:tetris-np-hard}.
On the other hand, Jelly-no-Puzzle, also by Qrostar, is known to be NP-hard is the general case~\cite{yan:j:jelly}.
Our paper uses NCL to study a game with sliding blocks and irreversible gravity, and
extends this line of work by providing a framework to study such games using only three gadgets in general, and by having only two gadgets when focusing on HANANO.

\section{Conclusion and Open Directions}\label{s:conclusion}
After establishing the NP-hardness of HANANO, Liu and Yang
stated as an open problem the task of determining whether $\hanano \in \np$. 
It follows from our PSPACE-completeness result that 
$\hanano \not\in \np$ unless $\np = \pspace$.

Another contribution of this paper is the use of the visibility representations.
Our Proposition~\ref{p:constrained_blue} and our 
Lemmas~\ref{l:or-equivalence} and~\ref{l:and-equivalence}
helped significantly reduce the number of gadgets needed.
Since these proofs are independent of 
HANANO, we believe that they can be reused to analyze additional games.

By leveraging schemas and symmetry, we only needed to provide three gadgets (instead of 
32 gadgets).
If we focus our attention to HANANO, then we can derive the red bend gadget from $\booobb$. And so, an interesting direction would be to investigate whether the reduction
could be carried out using only two gadgets (or even one) in the general case. (One might posit that by placing a constrained blue edge terminator on the blue edge for $\brr$, we get what looks like a red bend gadget, but that omits the inflow constraint of two. Thus both red edges would need to face into the gadget, which is too strong of a requirement.)
Finally, we mention that
our movable gray blocks have up to six sides in two gadgets and up to eight sides in the third gadget. 
We followed the closest paper in the literature,
i.e., that of Liu and Yang~\cite{liu-yan:j:hanano}, which states that 
those movable gray blocks can have ``any size or shape''~\cite{liu-yan:j:hanano}. However,
it would be interesting if our result could be strengthened to only have movable gray blocks
with exactly four sides.

\section*{Acknowledgments}  
We thank
Benjamin Carleton,
Lane A. Hemaspaandra,
David E. Narv\'{a}ez,
Conor Taliancich, and 
Henry B. Welles, and the anonymous reviewers
for their helpful comments and suggestions.

\bibliographystyle{alpha}
\bibliography{exported}

\newcommand{\etalchar}[1]{$^{#1}$}
\begin{thebibliography}{BDD{\etalchar{+}}02}

\bibitem[ABD{\etalchar{+}}20]{ani-bos-dem-dio-hen-lyn:c:door-pspace-hard}
J.~Ani, J.~Bosboom, E.~Demaine, Y.~Diomidov, D.~Hendrickson, and J.~Lynch.
\newblock Walking through doors is hard, even without staircases: {Proving}
  {PSPACE}-hardness via planar assemblies of door gadgets.
\newblock In {\em Proceedings of the 10th International Conference on Fun with
  Algorithms}, volume 157, pages 3:1--3:23, 2020.

\bibitem[ACD{\etalchar{+}}20]{asi-cou-dem-dem-hes-lyn-sin:j:tetris-np-hard}
S.~Asif, M.~Coulombe, E.~Demaine, M.~Demaine, A.~Hesterberg, J.~Lynch, and
  M.~Singhal.
\newblock Tetris is {NP}-hard even with ${O}$(1) rows or columns.
\newblock {\em Journal of Information Processing}, 28:942--958, 2020.

\bibitem[BDD{\etalchar{+}}02]{bie-dem-dem-fle-jac-mun:b:clickomania}
T.~Biedl, E.~Demaine, M.~Demaine, R.~Fleischer, L.~Jacobson, and J.~I. Munro.
\newblock The complexity of {Clickomania}.
\newblock In R.~J. Nowakowski, editor, {\em More Games of No Chance}, pages
  389--404. Cambridge University Press, Cambridge, England, 2002.

\bibitem[BDH{\etalchar{+}}04]{bre-dem-hoh-hoo-kos-lib:j:tetris-hard}
R.~Breukelaar, E.~Demaine, S.~Hohenberger, H.~Hoogeboom, W.~Kosters, and
  D.~Liben-Nowell.
\newblock Tetris is hard, even to approximate.
\newblock {\em International Journal of Computational Geometry \&
  Applications}, 14(1--2):41--68, 2004.

\bibitem[DVW16]{dem-vig-wil:c:mario-bros}
E.~Demaine, G.~Viglietta, and A.~Williams.
\newblock Super {Mario} {Bros}. is harder/easier than we thought.
\newblock In {\em Proceedings of the 9th International Conference on Fun with
  Algorithms}, pages 13:1--13:14, 2016.

\bibitem[Fri01]{fri:a:cubic}
E.~Friedman.
\newblock Cubic is {NP}-complete.
\newblock Presented at the 34th Annual {Florida} {MAA} Section Meeting, 2001.

\bibitem[HD05]{dem-hea:j:ncl-pspace}
R.~Hearn and E.~Demaine.
\newblock {PSPACE}-completeness of sliding-block puzzles and other problems
  through the nondeterministic constraint logic model of computation.
\newblock {\em Theoretical Computer Science}, 343(1-2):72--96, 2005.

\bibitem[HD09]{dem-hea:b:games-puzzles-comp}
R.~Hearn and E.~Demaine.
\newblock {\em Games, Puzzles, and Computation}.
\newblock CRC Press, 2009.

\bibitem[LY19]{liu-yan:j:hanano}
Z.~Liu and C.~Yang.
\newblock Hanano {Puzzle} is {NP}-hard.
\newblock {\em Information Processing Letters}, 145:6--10, 2019.

\bibitem[Qro11]{qro:w:hanano}
Qrostar.
\newblock Hanano {Puzzle}.
\newblock \url{https://qrostar.skr.jp/en/hanano/}, 2011.

\bibitem[Qro22]{qro:perscomm:hanano}
Qrostar, 2022.
\newblock Personal communication.

\bibitem[Tam16]{tam:b:graph-drawing}
R.~Tamassia.
\newblock {\em Handbook of Graph Drawing and Visualization}.
\newblock Chapman and Hall/CRC, 2016.

\bibitem[TT86]{tam-tol:j:visibility-representations}
R.~Tamassia and I.~Tollis.
\newblock A unified approach to visibility representations of planar graphs.
\newblock {\em Discrete \& Computational Geometry}, 1:321--341, 1986.

\bibitem[Yan18]{yan:j:jelly}
C.~Yang.
\newblock On the complexity of {Jelly}-no-{Puzzle}.
\newblock In {\em Japanese Conference on Discrete and Computational Geometry,
  Graphs, and Games}, pages 165--174, 2018.

\end{thebibliography}

\clearpage

\section*{Appendix}
\appendix

\section{Example of a (Simplified) Constrained Blue Terminator Gadget}

It is easy to see that in Figure~\ref{f:constrained_blue} is equivalent to what we call a constrained
blue edge terminator gadget. Indeed, we can view B1 has indicating that an edge is pointing into this gadget,
and if B1 were to leave the gadget, then G1 would fall, making it impossible for B2 to bloom.

\begin{figure}[h!]
	\centering
	\myfbox{0.35}{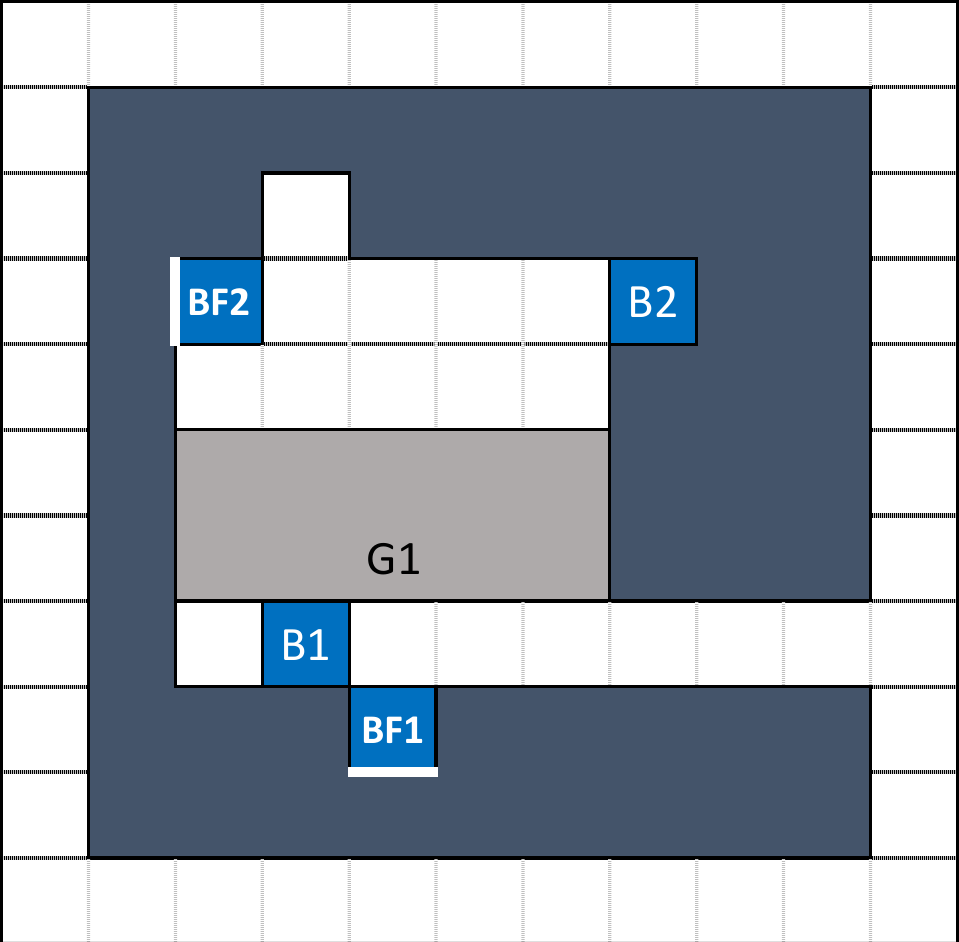}
	\caption{A simplified constrained blue edge terminator gadget.}\label{f:constrained_blue}
\end{figure}

\section{Sketch of a Reduction}\label{a:worked-example}
Let us consider as a working example, the graph in Figure~\ref{f:ncl-example}.
We already have a visibility representation from Figure~\ref{f:visibility-rep-example}.
Let the target edge to be flipped be $(C, B)$, i.e., edge~4.
Figure~\ref{f:replaced-pictorial} shows a sketch of the result.

\begin{figure}[ht!]
\centering
\includegraphics{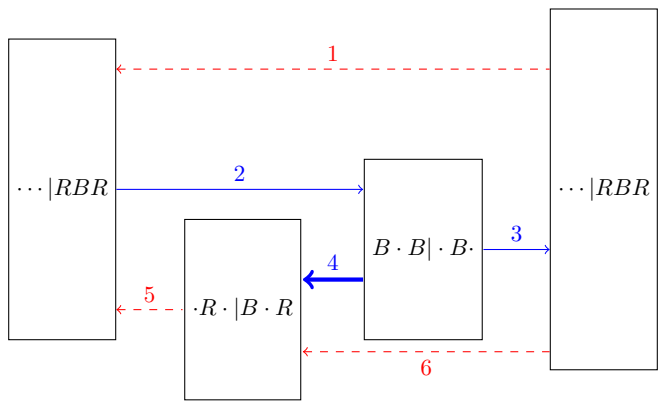}
\caption{Sketch of how the gadgets map onto the planar 
grid.}\label{f:replaced-pictorial}
\end{figure}

The horizontal lines now represent tunnels, the $\orobor$ gadget has one less
blue flower.
Empty spaces can be filled with immovable gray blocks.
The target edge has been boldfaced.

\end{document}